\documentclass[12pt]{iopart}

\usepackage{graphicx,subcaption}
\usepackage{natbib}
\usepackage[colorlinks,linkcolor=blue,citecolor=blue,urlcolor=blue]{hyperref}

\expandafter\let\csname equation*\endcsname\relax
\expandafter\let\csname endequation*\endcsname\relax

\usepackage{amsmath, amssymb}
\usepackage{amsthm}
\usepackage{threeparttable}

\newtheorem*{lemma}{Lemma}
\newtheorem*{theorem}{Theorem}

\newcommand{\red}[1]{{\color{red}{#1}}}
\renewcommand{\red}[1]{#1}

\newcommand{\eqsplit}[1]{
\begin{equation}
    \begin{split}
        #1
    \end{split}
\end{equation}
}

\newcommand{\tens}[1]{\mathsf{#1}}

\newcommand{\tA}{\tens{A}}

\renewcommand{\vec}[1]{\mathbf{#1}}

\begin{document}
\title[GNN-based Resource Allocation]{Graph Neural Network-based Resource Allocation Strategies for Multi-Object Spectroscopy}
\author{Tianshu Wang$^1$, Peter Melchior$^{1,2}$}
\address{$^1$Department of Astrophysical Sciences, Princeton University, Princeton, NJ 08544, USA}
\address{$^2$Center for Statistics and Machine Learning, Princeton University, Princeton, NJ 08544, USA}
\ead{tianshuw@princeton.edu; peter.melchior@princeton.edu }

\begin{abstract}
\red{Resource allocation problems are often approached with linear programming techniques. But many concrete allocation problems in the experimental and observational sciences cannot or should not be expressed in the form of linear objective functions.
Even if the objective is linear, its parameters may not be known beforehand because they depend on the results of the experiment for which the allocation is to be determined.}
To address these challenges, we present a bipartite Graph Neural Network architecture for trainable resource allocation strategies. 
Items of value and constraints form the two sets of graph nodes, which are connected by edges corresponding to possible allocations.
The GNN is trained on simulations or past problem occurrences to maximize any user-supplied, scientifically motivated objective function, augmented by an infeasibility penalty. The amount of feasibility violation can be tuned in relation to any available slack in the system.
We apply this method to optimize the astronomical target selection strategy for the highly multiplexed Subaru Prime Focus Spectrograph instrument, where it \red{shows superior results to direct gradient descent optimization and extends the capabilities of the currently employed solver which uses linear objective functions}.
The development of this method enables fast adjustment and deployment of allocation strategies, statistical analyses of allocation patterns, and fully differentiable, science-driven solutions for resource allocation problems.
\end{abstract}

\vspace{2pc}
\noindent{\it Keywords}:

%
\submitto{Machine Learning: Science and Technology}
%
%
%

\section{Introduction}
\label{sec:intro}
Resource allocation deals with the distribution of a fixed amount of resources through a number of admissible actions so as to minimize the incurred cost or maximize the resulting utility. The problem is encountered in a variety of application areas, including load distribution, production planning, computer resource allocation, queuing control, portfolio selection, and apportionment \citep{Katoh1998}. We are particularly interested in allocation problems arising in astronomical research, where the resource to be allocated is observing time at specific telescopes that are expensive to operate, and the utility is given by the scientific information gained from the chosen set of observations. With improved resource allocation strategies, astronomers can expect larger scientific yields or lower operational costs. The challenge lies in the large number of celestial objects that could in principle be observed and the large number of instrumental configurations that could be chosen.


\red{Powerful optimization packages for constrained and mixed-integer optimization like \texttt{GUROBI} can be employed to solve allocation problems, e.g. as a minimum-cost maximum flow-problem \citep{Bertsekas1998-ea}. But this approach has several limitations.
First, the fastest algorithms require a linear programming (LP) formulation, i.e. one in which the objective function and the constraints are linear in the allocations: $f(x) = c^\top x$, subject to $\mathbf{A}x\leq b$.
Although many problems can be expressed as LP, it does not permit cases in which different allocations interact or interfere with each other, and we will show that such cases can easily arise.} 
Second, the minimizers of the respective objective functions are themselves not differentiable with respect to the parameters of the problem such as per-item costs $c$. This precludes such an approach in situations where the actual cost structure is not known a priori, as is often the case in scientific settings.
This limitation can be overcome by treating the non-differentiable solver as a component in an extended gradient-based optimization \citep{Amos2017-ot, cvxpylayers2019, vlastelica2020, Donti2021-rp}, at the expense of additional hyper-parameters and the cost of running the solver inside of the optimization loop. 
Third, increasingly accurate analyses in astronomy and cosmology demand very detailed modeling of the processes that define the set of ultimately observed celestial objects \citep{Rix2021-nv}. It has thus become commonplace to perform hundreds or thousands of simulations to determine the actual ``selection function'' of the observing program \citep[e.g.][]{Ross2017-xg, Mints2019-pi, Everett2020-ya}. Complex MIP solvers, run either directly or as component of a deep learning architecture, would constitute a computational bottleneck for these efforts.

In this paper we present a Graph Neural Network (GNN) solver for general resource allocation problems.
The underlying bipartite graph comprises sets of items and constraints as nodes, connected by edges representing possible allocations.
The GNN is trained on simulations or past problem instances to learn how to take actions, i.e. to assign allocations, that satisfy all constraints within the posed resource limits while maximizing a user-supplied utility function.
In contrast to reinforcement learning, we do not assign immediate rewards for specific actions.
\red{We also do not solve the assignment or scheduling problem, i.e. to determine a specific feasible sequence of assignments to maximize a given objective in a multi-epoch observing program.}
Instead, the GNN predicts the \emph{amount} of resources to allocate for every object such that there is at least one feasible sequence.
We have recently demonstrated that GNNs with such a continuous relaxation solve allocation problems better than strong human heuristics or parameterized evolutionary strategies even if the utility function can only be learned by interacting with the environment \citep{Cranmer2021-jd}.
What we show here is that bipartite GNNs can efficiently learn to obey feasibility constraints of complex real-world environments with discrete allocations.

The remainder of this paper is structured as follows:
In \autoref{sec:method} we describe the problem definition and our GNN solver in detail. In \autoref{sec:apply} we specialize this method to two concrete examples of selecting the optimal set of galaxies to observe with the upcoming Prime Focus Spectrograph, a highly multiplexed instrument on the Subaru Telescope, located on Maunakea in Hawai`i, USA.
In \autoref{sec:training} we discuss training and initialization, and in \autoref{sec:result} we compare the results of our GNN to \red{those of direct gradient descent and the currently established baseline from a LP solver}. We conclude in \autoref{sec:summary} with a summary and an outlook of possible extensions of our approach.

\section{Methodology}
\label{sec:method}

\subsection{Problem Definition}

Following \cite{Katoh1998} and \cite{Kurt1995}, the general resource allocation problem has the form of a (non-)linear programming problem, where we seek to

\begin{equation}
\label{equ:general}
\begin{array}{ll}
 \mathrm{maximize} f(x_1,\dots,x_J) & \mathrm{subject\ to
}\\
h_k(x_1,\dots,x_J)\leq 0 & k\in\{1,\dots,K_\mathrm{ineq}\}\\
h_k(x_1,\dots,x_J)=0 & k\in\{K_\mathrm{ineq}+1,\dots,K_\mathrm{ineq}+K_\mathrm{eq}\}.
\end{array}
\end{equation}
The objective function $f$ depends on allocations $x_j$ $(j \in \{1,\dots,J\})$ that can be either discrete, $x_j \in \{0,1,2,...,T_\mathrm{max}\}$, or continuous, $x_j\in[0,T_\mathrm{max}]$, up to for some finite $T_\mathrm{max}$.
Constraint equations $h_k$ $(k \in \{1,\dots,K=K_\mathrm{ineq}+K_\mathrm{eq}\})$ limit the configurations under which these allocation can be distributed.
Depending on the features of the objective function and the types of  constraints, resource allocation problems form different classes. Cases where the objective function or constraints are linear or convex have known solutions (e.g., \cite{Federgruen1986,Kurt1995,Katoh1998,Shi2015}). But resource allocation problems remain conceptually challenging when the objective function or constraints have more complicated forms, and numerically demanding when allocations are discrete and when the number of variables is large.

We find it beneficial to reparameterize the objective function, i.e. we seek to
\begin{equation}
\label{equ:resource_alloc}
\begin{array}{ll}
\mathrm{maximize} f(y_1\dots,y_I) & \mathrm{subject\ to
}\\
y_i=g_i(x_1,\dots x_J) & i\in\{1,\dots,I\}\\
h_k(x_1,\dots,x_J)\leq 0 & k\in\{1,\dots,K_\mathrm{ineq}\}\\
h_k(x_1,\dots,x_J)=0 & k\in\{K_\mathrm{ineq}+1,\dots,K_\mathrm{ineq}+K_\mathrm{eq}\},
\end{array}
\end{equation}
by means of functions $g_i$ $(i=1,\dots,I)$.
The motivation behind the reparameterization lies in symmetries of the objective function which often permit a strong compression from the full set of $J$ allocations to a much smaller number of variables $y_i$. In particular, if the objective function only depends on the total allocation (e.g. in the single knapsack problem), a single $y_1= \sum_{j=1}^J x_j$ suffices for any $J$.
For resource allocation problems, the set of $y$'s correspond to the items of value for which the allocations are made.

\autoref{equ:general} and \autoref{equ:resource_alloc} can represent many types of optimization problems. What makes resource allocation problems special is that their
$h$ and $g$ functions are permutation invariant,
i.e. there exists functions $\rho$ and $\phi$ such that e.g. $h(x_1,\dots,x_J)=\rho\left(\sum_j \phi(x_j)\right)$ \citep{Zaheer2017-jg}.
Consequently, constraint and the item functions do not depend on the order of arguments. 

\subsection{Graph Construction}

According to \autoref{equ:resource_alloc}, the set of allocations $x_1,\dots,x_J$ provides the arguments to both the $g$ and the $h$ functions.
This dependency structure suggest a representation of the allocation problem in the form of a bipartite graph, where one set of nodes represent the constraints $h_k$ $(k=1,\dots,K)$ and the other represents the items $g_i$ $(i=1,\dots,I)$. Whenever a particular $x_j$ appears as argument of the nodes $g_i$ and $h_k$, the graph has an edge connecting these two nodes.
The set of allocations $x_j$ $(j=1,\dots,J)$ thus defines the connectivity of the graph, with any individual $x_j$ potentially being represented by multiple edges.
Because of the suitable representation, bipartite graphs have a long history in assignment and allocation problems \citep[e.g.][]{Bertsekas1998-ea, Wong2007-fk, Abanto-Leon2017-nb, Nair2020-tv}.

Of particular relevance for this work is that the constraints and items form two classes of similar, permutation invariant functions, as we demonstrate with the following example.

\subsection{Example: Multiple Knapsack Problem}
We demonstrate the ansatz above for the 0-1 Multiple Knapsack Problem (MKP). Given a set of $I$ items and a set of $K$ knapsacks, with $v_i$ and $w_i$ being the value and weight of item $i$, and $c_k$ the capacity of knapsack $k$, the task is to select $K$ disjoint subsets of items such that they maximize the total value. Each subset is assigned to a different knapsack, whose capacity cannot be less than the total weight of items in the subset, i.e. we seek 
\eqsplit{
\label{eq:knapsack}
&\underset{\{x_{11},\dots,x_{IK}\}}{\mathrm{argmax}}\ \sum_{k=1}^K\sum_{i=1}^I v_i x_{ik}\\
&\forall k: \sum_{i=1}^I w_i x_{ik}-c_k \leq 0\\
&\forall i: \sum_{k=1}^K x_{ik}-1 \leq 0\\
&\forall k,i: x_{ik}\in\{0,1\}.
}
Although there are $I\times K$ allocation variables, the objective function actually only depends on $I$ independent combinations of them: $\sum_{k=1}^K\sum_{i=1}^I v_i x_{ik}=\sum_{i=1}^I v_i y_i$, where $y_i=\sum_{k=1}^K x_{ik}$. We can effectively combine the per-item constraints with the definition of the $y_i$ by defining itemization functions $g_i(x)=\min(\sum_{k=1}^K x_{ik},1)$.
The MKP can be simplified and written in the form of \autoref{equ:resource_alloc}:
\eqsplit{
&\underset{\{x_{11},\dots,x_{IK}\}}{\mathrm{argmax}}\ \sum_{i=1}^I v_i y_i\\
&\forall i: y_i=g_i(x)=\min(\sum_{k=1}^K x_{ik},1)\\
&\forall k: h_k= \sum_{i=1}^I w_i x_{ik}-c_k\leq0\\
&x_{ik}\in\{0,1\}.
\label{equ:MKP}
}
The maximizers of \autoref{equ:MKP} are equivalent to those of \autoref{eq:knapsack} with respect to the objective function. The latter formulation permits unfeasible assignments of a single item to multiple knapsacks, which can be corrected by a single pass over all items and removal of all but one assigned knapsack.

From this formulation, we construct a graph as follows: Each $g_i$ is one item node, and each $h_k$ is one constraint node.
The edges $x_{ik}$ connect both sets of nodes and form a complete bipartite graph.
Because the MKP has one constraint equation per knapsack, $h$-nodes represent the knapsacks and the $g$-nodes the items. It is evident that the underlying functions are structurally similar and permutation invariant.

This construction is similar to the graph representation of a MIP in \citet{Nair2020-tv}, but not identical.
They restrict their problem to objectives of the form $\sum_j c_j x_j$ and directly identify the item nodes with $x_j$, we allow for arbitrary permutation invariant functions $g$ to modify the relation between $x_j$ and $y_i$ in \autoref{equ:resource_alloc}.
Also, in the graph the edges correspond to the elements $a_{ij}$ of the matrix in the linear constraint equation $Ax\leq b$, i.e. the carry information about feasibility, whereas the edges in our graph carry information about the allocation amount. 

\subsection{GNN Definition}
\label{sec:block}

\begin{figure}[t]
    \centering
    \begin{subfigure}[t]{.42\textwidth}
    \includegraphics[width=1.\textwidth]{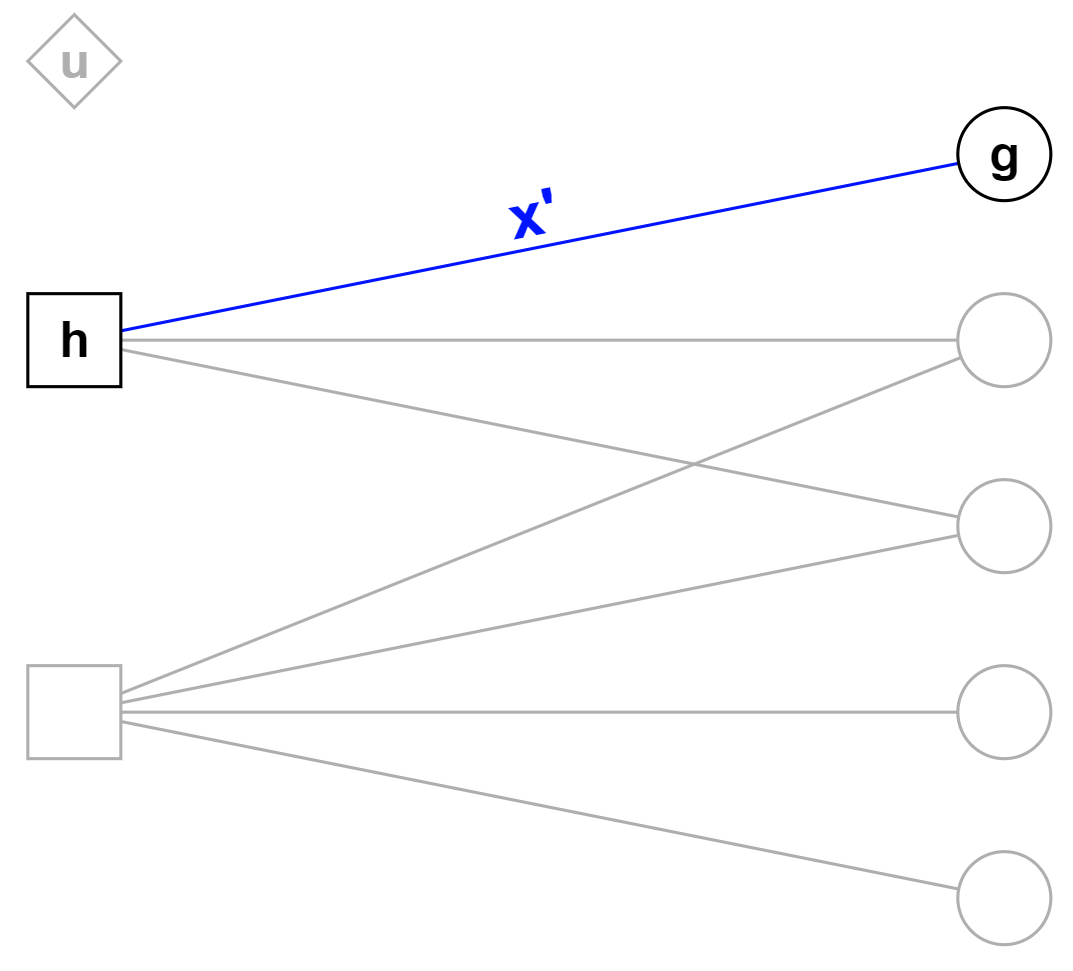}
    \caption{Edge update given the node features.}
    \end{subfigure}
    \begin{subfigure}[t]{.42\textwidth}
    \includegraphics[width=1.\textwidth]{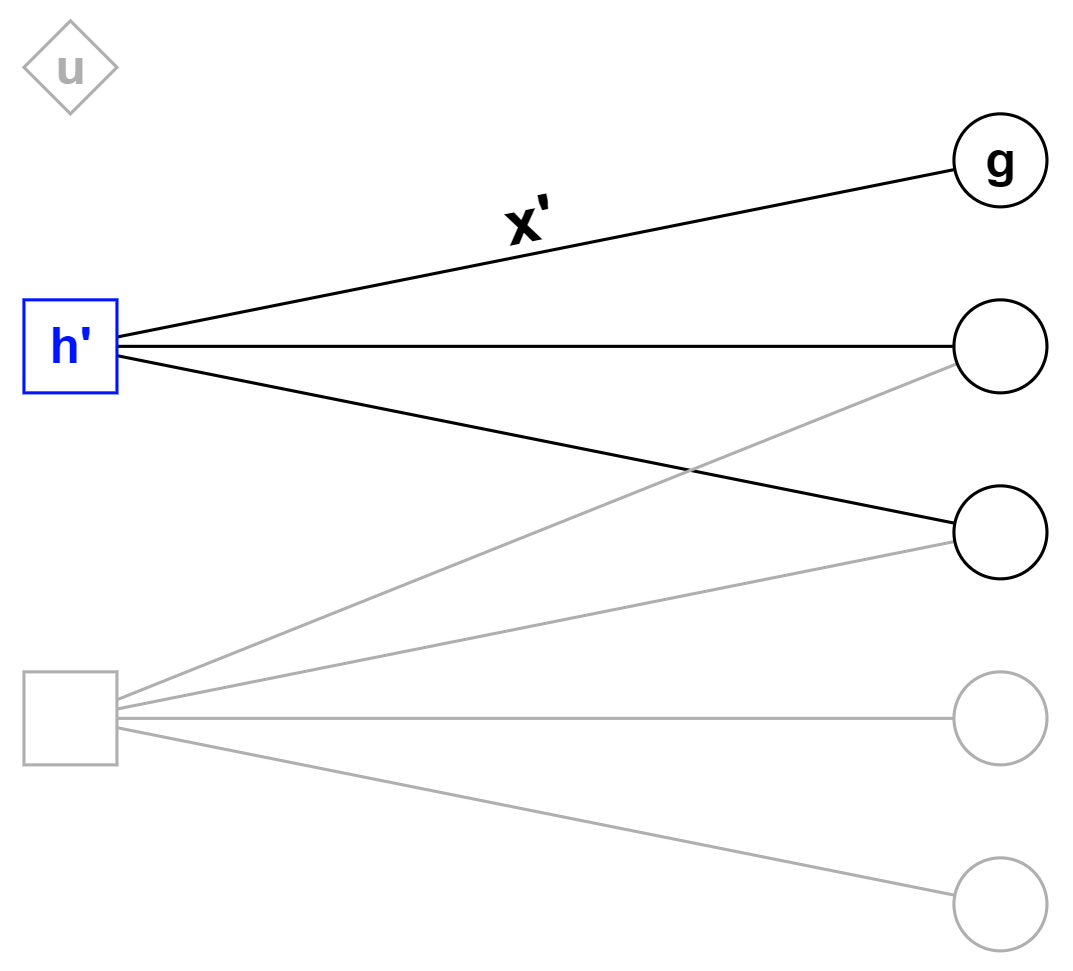}
    \caption{Constraint node update given edge features and connected item node features}
    \end{subfigure}
    \begin{subfigure}[t]{.42\textwidth}
    \includegraphics[width=1.\textwidth]{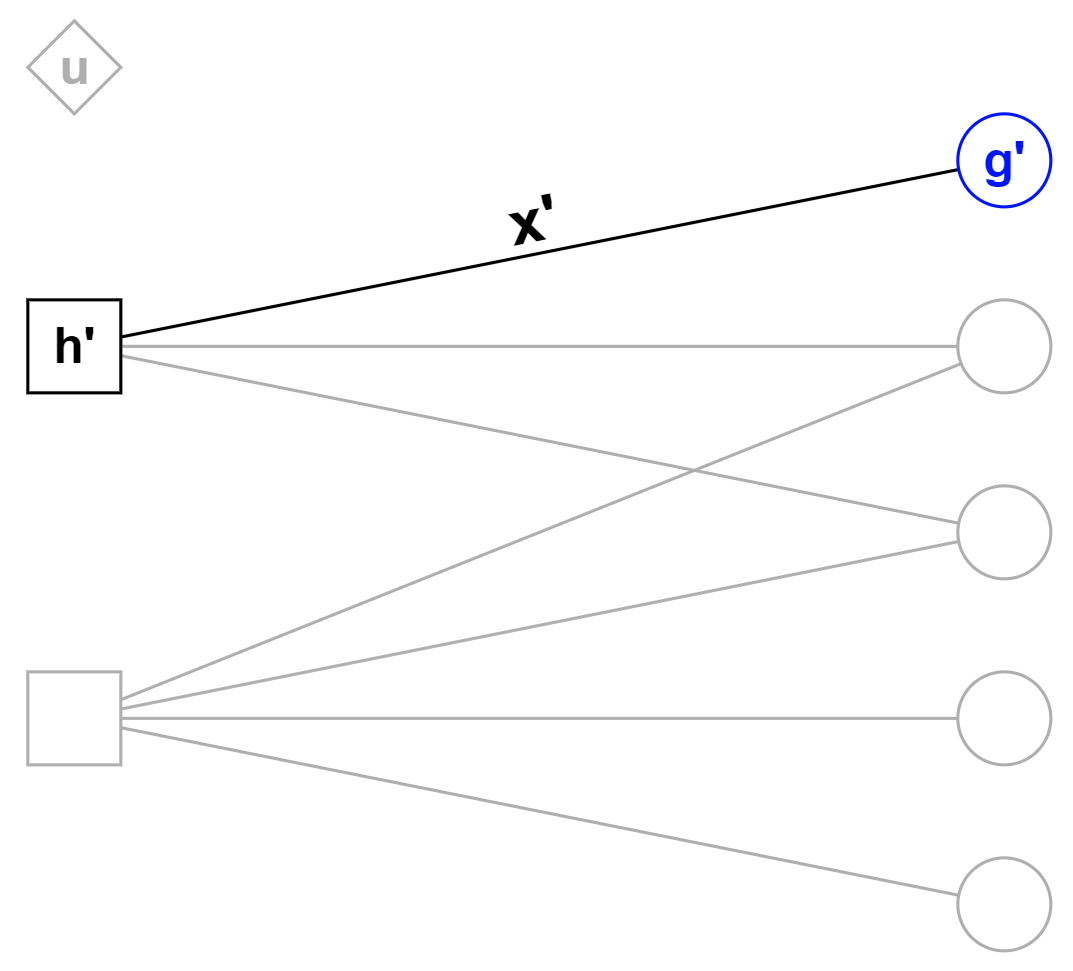}
    \caption{Item node update given edge features and connected constraint node features}
    \end{subfigure}
    \begin{subfigure}[t]{.42\textwidth}
    \includegraphics[width=1.\textwidth]{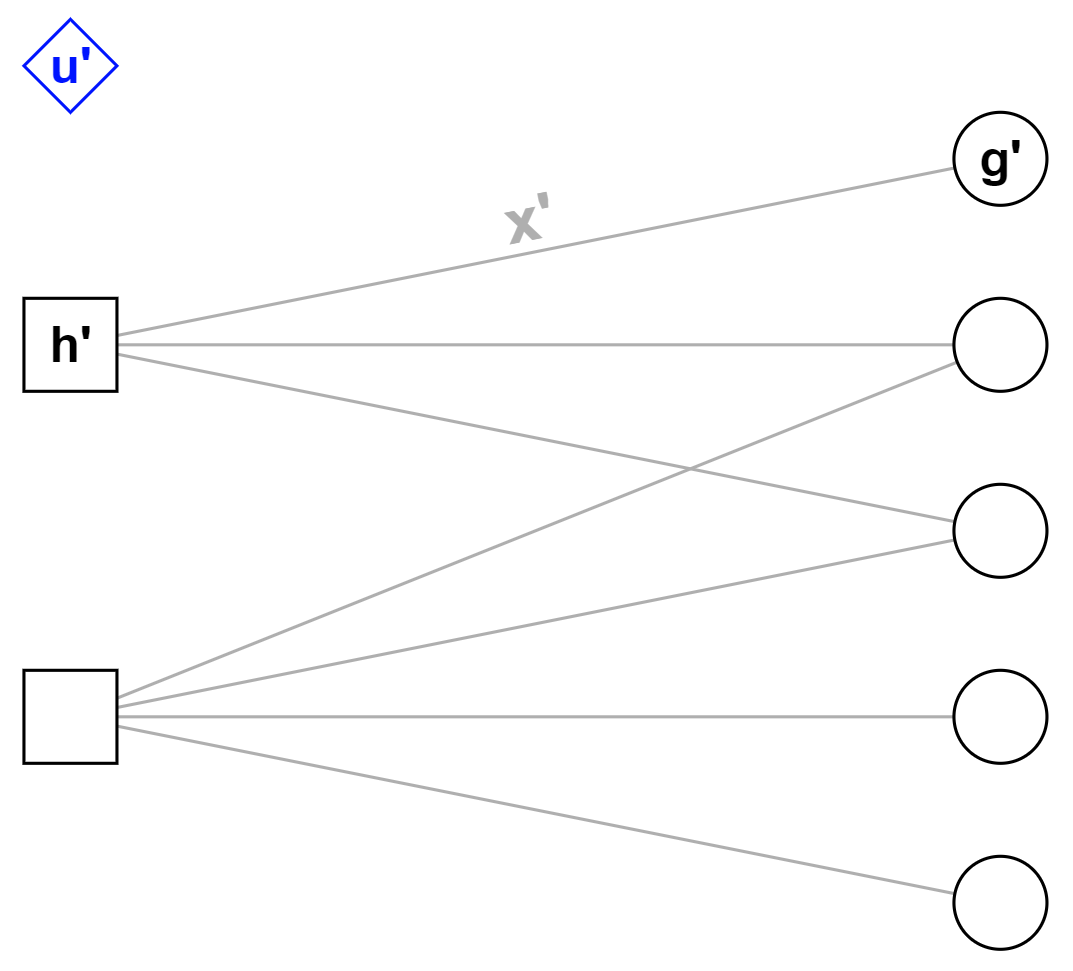}
    \caption{Global update given node features.}
    \end{subfigure}
    \caption{Updates in a GNN block. Blue shows the element that is being updated, black indicates the elements that are involved in the update and grey elements are unused. $h$ and $g$ represent the attributes of the two types of nodes in the bipartite graph, $x$ represents the edge attributes, and $u$ is the global attributes of the graph. Parameters with primes are the updated values.}
    \label{fig:block}
\end{figure}

\red{Unlike traditional MIP solvers, or their neural reformulation \citep{Nair2020-tv}}, we seek to find solutions where the parameters of the problem are not fully determined. For the MKP that can arise e.g. when item values are not known a priori.
In addition, we seek an architecture that learns to solve a particular kind of allocation problem rather than running an explicit solver for every instance of the problem as proposed in e.g. \citet{vlastelica2020}. 
The expected performance gains are important for statistical assessments of the probability of particular allocations.
We thus want to describe allocation problems with a differentiable, trainable model.

The $h$ and $g$ functions in allocation problems form two classes of similar functions, which means that we need to parameterize only the behavior of the classes, not of every class element. This allows us to model relations on the graph with a bipartite version of the GNN blocks defined in \cite{Battaglia18}. 
Specifically, the bipartite GNN block has two distinct node models, instead of only one for the regular GNN block.
Both node models depend on their attached edges and corresponding nodes, and the edge model depends on both sets of attached nodes, whose features we simply concatenate.%
\footnote{Generalizations to tripartite or even more complex graphs are conceivable to address problems in which the constraint and item functions cannot be represented by only two classes.}

In addition to the graph connectivity, each of the three types of models needs to access auxiliary features, such as the item weights in the MKP, so we make sure that each element in the graph has direct access to all information related to its role in the optimization problem (see \autoref{sec:training} for concrete examples).
We hypothesize that the competing demands on available resources can better be met when each node model has access not only to the edge features, but also to the node features on the opposite side of the edge. We therefore concatenate them into an extended edge feature set, expecting that this renders message passing more efficient and thus reduces the number of GNN blocks.
Formally, let $n_x, n_h, n_g, n_u$ be the number of features carried by each edge, constraint node, item node, and global node, respectively.
Also, let $n^a_g$ and $n^a_h$ be the number of different aggregators of the item and constraint models to summarize the information carried by the (extended) edge features. We normally use four aggregators, namely the element-wise mean, variance, skewness, and kurtosis of the  edge features, unless the number of edges is too small to define some of the high-order moments.
Defining $\phi:\mathbb{R}^{(\cdot)}\rightarrow \mathbb{R}^{(\cdot)}$ as a multi-layer perceptron (MLP), our GNN block is thus comprised of $\{\phi^x,\phi^h,\phi^g,\phi^u\}$, where
\begin{itemize}
    \item $\phi^x:\mathbb{R}^{(n_x+n_h+n_g+n_u)}\rightarrow\mathbb{R}^{n_x}$ updates the edge features using the previous edge features, features from the two nodes connected to the edge, and global features;
    \item $\phi^h:\mathbb{R}^{(n_h+n^a_h(n_x+n_g)+1+n_u)}\rightarrow\mathbb{R}^{n_h}$ updates the constraint node features using the previous constraint node features, $n^a_h=4$ aggregators (element-wise mean, variance, skewness, and kurtosis) of the extended edge features, the number of connected edges, and the global features;
    \item $\phi^g:\mathbb{R}^{(n_g+n^a_g(n_x+n_h)+n_u)}\rightarrow\mathbb{R}^{n_g}$ updates the item node features using the previous item node features, the aggregated edge features, and the global features;
    \item $\phi^u:\mathbb{R}^{(n_h+n_g+n_u)}\rightarrow\mathbb{R}^{n_u}$ updates the global features using the mean of the node features and the previous global features.
\end{itemize}
The update sequence is built in a similar way as the \texttt{MetaLayer} class in the \texttt{PyGeometric} package \citep{fey2019}.
In particular, we place another MLP before the aggregation step, which renders the models more flexible, and is the reason why we can handle permutation invariant functions by $\phi^h$ and $\phi^g$ instead of merely symmetric functions \citep{Zaheer2017-jg}.
The updates proceed in the order of \autoref{fig:block}: first the edge model given the node features, then both node models given the respective edge features, and then a global model given the node features.


We stack 4 GNN blocks and perform batch normalization on all nodes and edge features, where the batch dimension is given by the number of nodes or edges of the graph. The number of GNN blocks depends on the complexity of the problem, with more blocks corresponding to more message-passing steps to negotiate between the competing demands on the minimizer of \autoref{eq:lagrangian}. Like \cite{Cranmer2021-jd} we find that 3 or 4 blocks suffice, and we leave determining the optimal number of blocks to forthcoming work.

The output of $\phi^x$ of the last GNN block is a real number $\tilde{x}_j$ and the corresponding $x_j$ is calculated by $x_j=T_\mathrm{max}\times\sigma(\tilde{x}_j)$. If the problem requires integer allocations, we apply a round function to the output.  During training, we replace the round function with the noisy sigmoid function \citep{Edward94}:
\eqsplit{
\label{eq:noisy-sigmoid}
&z\sim\mathcal{U}(-l/2,l/2)\\
&x'=x+z\\
&f(x)=\mathrm{floor}(x')+\sigma[k(x'-\tfrac{1}{2}-\mathrm{floor}(x'))],
}
where $k$ is the sharpness and $l$ is the noise level (see \autoref{fig:sigmoid}).

\begin{figure}[t]
    \centering
    \includegraphics[width=0.8\textwidth]{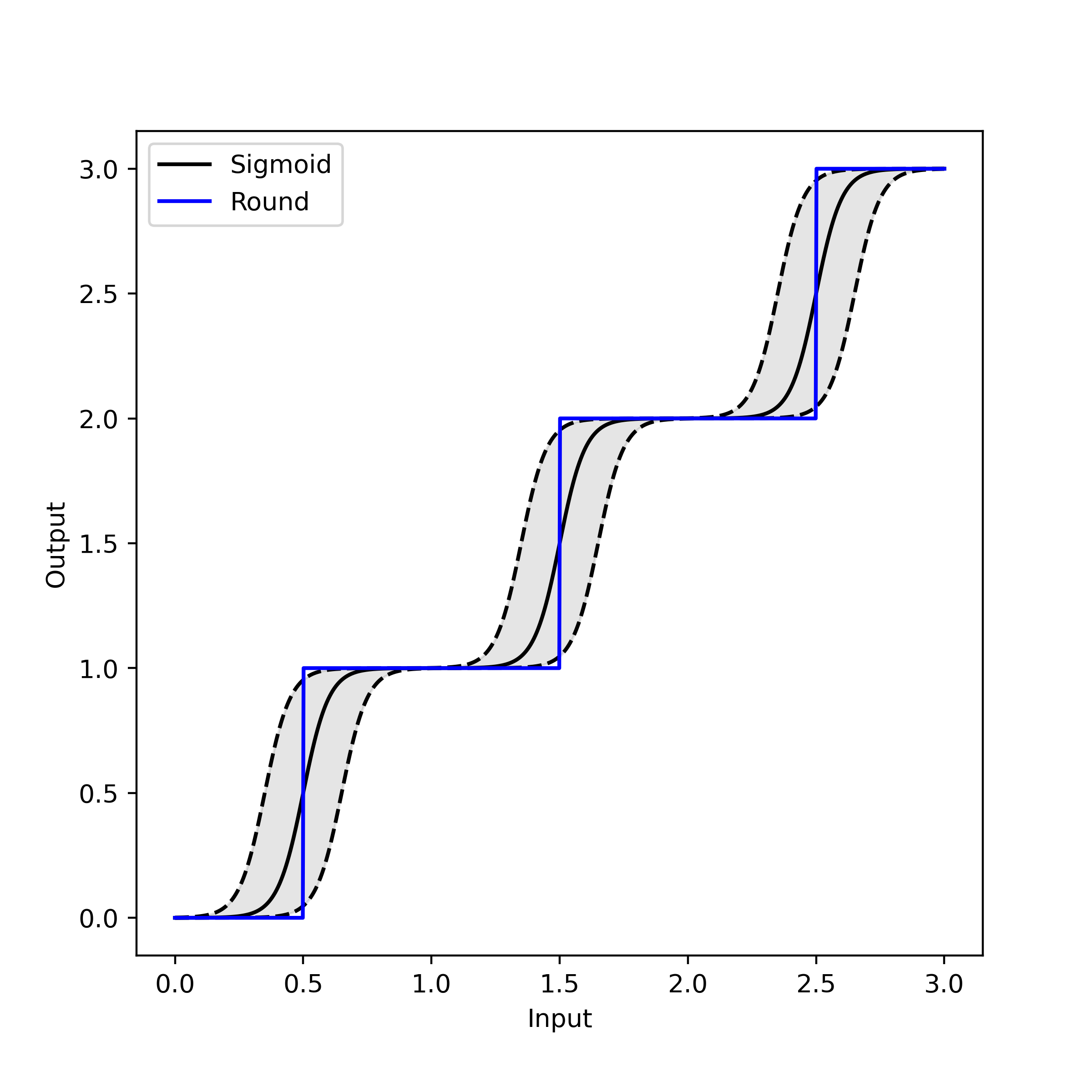}
    \caption{Comparison between noisy Sigmoid function and round function. Black curve is the exact Sigmoid function and the shadow shows the noise. Blue curve is the exact round function. In this figure, the sharpness is 20 and the noise level is 0.3}
    \label{fig:sigmoid}
\end{figure}

\subsection{Loss Function}
We define the loss function as the negative Lagrangian of \autoref{equ:resource_alloc},

\begin{equation}
\label{eq:lagrangian}
L(x_1,\dots,x_J) = -f(y_1,\dots,y_I)+\lambda\sum_{k=1}^K p_k\left[h_k(x_1,\dots,x_J)\right],
\end{equation}
where $y_i=g_i(x_1,\dots x_J)$ and $p_k$ are penalty functions appropriate for constraint violations, e.g. $\ell_1$, $\ell_2$ or ReLU.

The amount of penalty $\lambda > 0$ formally needs to be infinite if only feasible minimizers of \autoref{equ:resource_alloc} are accepted.
We relax this requirement by increasing the penalty to a large number during network training. 
Empirically, we find that this often leads to feasible solutions, or an amount of constraint violation that can tolerated due to slack in realistic settings. 
If solutions with exact feasibility are needed, one can make minor adjustments with a greedy algorithm, e.g. by removing the least valuable items in the case of overallocation.

\section{Application to the PFS Target Selection Problem}
\label{sec:apply}

\begin{figure}
    \centering
    \includegraphics{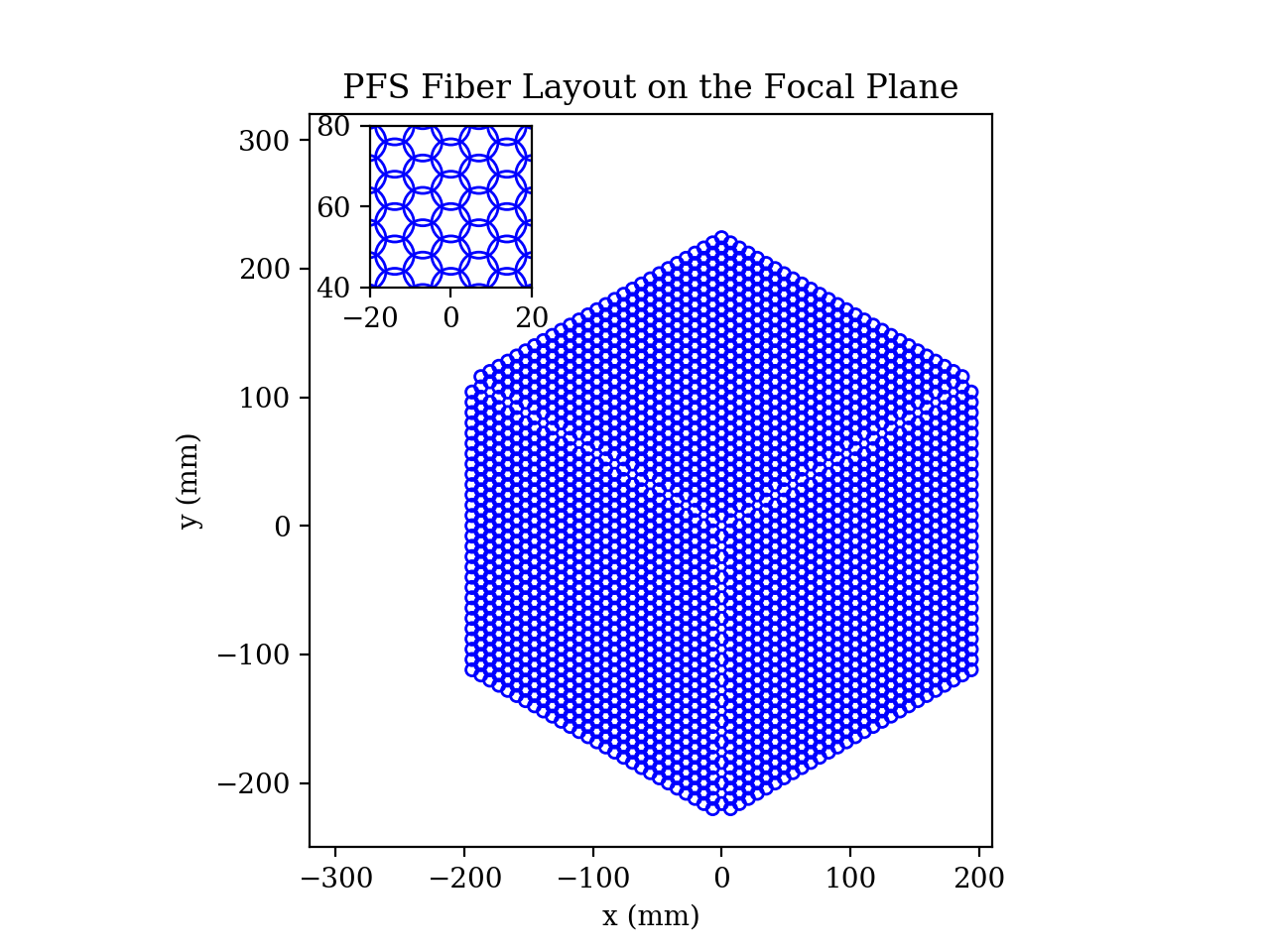}
    \caption{The fiber layout of the Prime Focus Spectrograph in focal-plane coordinates. Circles indicate the patrol region for all 2,394 fibers.}
    \label{fig:layout}
\end{figure}

The Prime Focus Spectrograph (PFS) is a wide-field, highly multiplexed optical and near-infrared spectrograph that will soon be installed at the 8.2m Subaru Telescope located at the peak of Maunakea in Hawai`i, USA \citep{Tamura16}. The instrument is equipped with 2,394 movable fibers distributed over a 1.3 deg$^2$ field of view.
The fibers can be moved laterally so that they can collect the light from astronomical objects they are pointed at. They stay in place for a configurable amount of time to feed light to the dispersive elements of the spectrograph, and ultimately to its detector, forming one `exposure'. 
Between exposures, every fiber can independently be positioned within a circle of $9.5\,\mathrm{mm}$ in diameter by an electro-mechanical actuator.
The whole fiber assembly is packed in a hexagonal pattern with $8\,\mathrm{mm}$ separation (see \autoref{fig:layout}). The overlap between adjacent `patrol regions` enables full sky coverage.%
\footnote{Because each galaxy can be reached by at most two fibers, we limit the aggregator in the item MLP $\phi^g$ to a simple element-wise sum, i.e. $n^a_g=1$. Using mean and variance does not yield any benefits, and higher-order moments would be ill-defined.}. 

\subsection{The Target Selection Problem}
Given a total time allocation budget $T$ and list of astronomical `targets' with their celestial positions and other characterizing features, a target selection strategy has to decide which targets to observe and, possibly for every single one, for how long.
\citet{Cranmer2021-jd} demonstrated that GNNs can solve this allocation problem, even with an implicit objective function, better than heuristics or simple parameterized strategies, but in their approach the allowed allocations $x_j\in[0,T_\mathrm{max}]$ were independent from each other, the only requirement being that $\sum_j x_j\leq T$.

For a multiplexed instrument such as the PFS the solutions are much more strongly constrained because the allocations for all 2,394 fibers in any given exposure must be identical. The allocations of different targets may differ by observing some targets more often than others.
For the planned PFS galaxy evolution program of this case study, each exposure time is fixed at 1 hours, with a total observing time budget, i.e. the sum of all exposure times, of $T=42\,\mathrm{h}$.

Specifically, let $I$ be the number of targets in a single field of view of the telescope. The objective function $f$ measures the scientific utility as a function of the time spent on each target, i.e., $f(\tau_1,\tau_2,...,\tau_I)$. It is related to the properties of the selected galaxies and the specific astrophysical questions at hand. In contrast to \cite{Cranmer2021-jd}, we demand for this work that $f$ is a known function.
\red{But unlike \citet{Lupton2002-dj} and \citet{Blanton2003-ni}, where $f$ is restricted to specific linear functions, we allow it to be an arbitrary permutation-invariant function.}
Let $\Phi_i$ ($1\leq i\leq I$) be the set of fibers that can reach the position of galaxy $i$, and let $\Psi_k$ ($1\leq k\leq 2394$) be the set of galaxies that fiber $k$ can reach. The total time spent on this field is $T$.
The target selection problem is then defined as the following optimization problem
\eqsplit{
\label{equ:opt}
&\arg\max_{t_{ik}}f(\tau_1,\tau_2,...,\tau_I)\\
&\forall i:\tau_i=\min(\sum_{k\in \Phi_i}t_{ik},T_\mathrm{max})\\
&\forall k: \sum_{i\in \Psi_k}t_{ik}\leq T\\
&\forall i,k: t_{ik}\in \{0,1,2,..,T_\mathrm{max}\}
}
where $t_{ik}$ is the time fiber $k$ spends on galaxy $i$, and we redefined all times in integer multiples of the base exposure time of $1\,\mathrm{h}$. The maximum time $T_\mathrm{max}\leq T$ any single galaxy can receive is set by the scientific program. Compared to the general form of resource allocation problem in \autoref{equ:resource_alloc}, we see that $\tau$'s and $t$'s correspond to $y$'s and $x$'s, respectively.

In principle, we must make sure that any fiber only observes at most one galaxy and any galaxy is observed by at most one fiber at each exposure:
\eqsplit{
\label{equ:decompose}
&\sum_{l=1}^Tt_{ikl}=t_{ik}\\
&\sum_{i\in\Psi_k}t_{ikl}\leq1\\
&\sum_{k\in\Phi_i}t_{ikl}\leq1\\
&t_{ikl}\in\{0,1\}
}
where $t_{ikl}$ is the time spent on target $i$ from fiber $k$ in exposure $l$. However, we prove in \autoref{sec:proof} that finding a sequence of exposure-level assignments is always possible as long as no explicitly sequence-dependent term appears in objective or constraint functions, and can be found in at most polynomial time.%
\footnote{If the relation between fibers and galaxies changes with time due to effects like dithering, i.e., $\Phi_i$ and $\Psi_k$ are time-dependent, we split one constraint into multiples  to make sure each constraint is time-invariant and  can be related to one fixed constraint node in the graph.}
Thus, we can focus on the optimization problem in \autoref{equ:opt} without having to worry about the sequence decomposition. 

The utility function can in some cases be written as the sum of the individual utilities of each galaxy, i.e. $f(\tau_1,\tau_2,...,\tau_I)=\sum_{i}f_i(\tau_i)$, which leads to a nonlinear MKP \red{that is already outside of the scope of LP solvers.}
However, the total scientific yield generally depends on the collective properties of \emph{all} observed galaxies.
For example, a scientific study may require that at least a certain number of galaxies be observed such that a combined measurement reaches a desired significance. The utility function is thus the sum of a separable part and a non-separable part, 
\begin{equation}
f(\tau_1,\tau_2,...,\tau_I)=\sum_{i}f_i(\tau_i)+s(\tau_1,\tau_2,...,\tau_I).
\end{equation}
We will define the specific form of $f$ for two cases below. The final loss function is then a specialization of \autoref{eq:lagrangian} for the problem in \autoref{equ:opt}:
\begin{equation}
L(t_{11},\dots,t_{IK})=-f(\tau_1,\tau_2,...,\tau_I)+\lambda\sum_{k=1}^{2394}p(\sum_{i\in \Psi_k}t_{ik}-T)
\end{equation}

\subsection{Case 1: Predefined Galaxy Classes}
\label{sec:case1}

The galaxy evolution program in the PFS Subaru Strategic Program (SSP) survey \citep{Takada2014-yr} currently plans to target a variety of galaxies, and has tentatively identified 12 distinct science cases and defined selection criteria for each of them. The sets of galaxies that satisfies these criteria define 12 galaxy classes.
Each of the science cases also defines the number of exposures a galaxy in the respective class should receive.
\autoref{tab:classes} shows the 12 galaxy classes and the number of galaxies satisfying the selection criteria in a reference field. The total number of visits available is $T=42$, while the time spent on a single galaxy is limited to $T_\mathrm{max}=15$.

A general goal in designing the criteria \red{and costs $c_m$} of \autoref{tab:classes} is that the program observes as many galaxies as possible for every science case, ideally in a reasonably equitable distribution.
We formalize this by means of an objective function that maximizes the minimal per-class completeness over all 12 classes:
\begin{equation}
\label{equ:case1-obj}
f(\tau_1,\dots,\tau_I)=\min\left(\frac{n_1}{N_1},\frac{n_2}{N_2},...,\frac{n_{12}}{N_{12}}\right)\ \mathrm{with}\ n_m\equiv\sum_{i\in\Theta_m}\sigma\left(\frac{\tau_i+0.5-T_m}{0.2}\right),
\end{equation}
where $\Theta_m$ is the $m$-th class, $T_m$ is its proposed per-galaxy exposure time, and $N_m$ is the number of galaxies in the field falling into class $\Theta_m$. 
We denote $n_m$ as the number of fully observed galaxies in class $m$.
During training we use a sigmoid function, as indicated above, to smoothly approximate the step function, but at test time we replace it with the actual step function to count distinct allocations.
We chose as penalty function $p$ a squared ReLU function, i.e. an inequality constraint on the fiber allocation capacity, because there is no need to exhaust all resources if no gain in $f$ is achieved.

\red{\autoref{equ:case1-obj} is evidently non-linear and entirely non-separable as it seeks to balance the allocation across 12 classes, each of which is comprised of thousands of galaxies, for which multi-exposure allocations need to be determined. The exact form of the equation could be chosen differently, but the underlying idea is motivated by the current survey design principle for the PFS galaxy evolution program.}

\begin{table*}[]
\caption{Predefined Galaxy Classes for Case 1. The required exposure times $T_m$ have been determined by the PFS Galaxy Evolution program on the basis of expected performance of the instrument, \red{and the costs $c_m$ provide the current baseline, which has been found by manual exploration}. $N_m$ denotes the number of galaxies that satisfy the class selection criteria in a reference field.}
\centering
\label{tab:classes}
\begin{tabular}{lcccccc}
\hline\hline
&1&2&3&4&5&6\\
\hline
$T_m$ [h] &2&2&2&12&6&6\\
$c_m$ &19683&19683&59049&531441&177147&177147\\
$N_m$ [$10^3$] &68.2&69.3&96.3&14.4&22.0&8.3\\
\hline\hline\\
\hline\hline
&7&8&9&10&11&12\\
\hline
$T_m$ [h] &12&6&3&6&12&1--15$^a$\\
$c_m$ &531441&177147&59049&177147&531441&59049\\
$N_m$ [$10^3$] &14.0&22.0&7.4&4.5&2.8&9.7\\
\hline\hline
\end{tabular}
\footnotesize
\\$^a$Each galaxy in this class has an independent exposure time requirement. 
\end{table*}

\subsection{Case 2: A General Objective Function}
\label{sec:case2}

For case 2, we envision a smaller observing program that could be carried out with PFS in a single night.
We thus adopt a very modest allocation of $T=6\,\mathrm{h}$ and $T_\mathrm{max}=4\,\mathrm{h}$.
Instead of adopting predefined classes, we now combine two objectives: 1) maximizing the number of galaxies, for which spectroscopic redshifts $z$ can be determined with a precision $\delta z<0.001$. Such a sample of galaxies can be used to reconstruct the so-called cosmic web \citep[e.g.][]{Jasche2015-qi, Horowitz2021-du}.
2) creating a sample of at least 5,000 faint galaxies at relatively large redshift $z>1$ and within a range masses, $11.8<\log_{10}M_\mathrm{halo}<12.5$, that should be observed at least once. The purpose of such a sample is to aggregate their spectra and achieve high signal-to-noise ratio to test for the presence of specific spectral features \citep[e.g][]{Carnall2019-rc, Salvador-Rusinol2019-ax}.
While the specific definitions of these objectives are hypothetical, they serve as an example of a directly science-driven fiber allocation strategy for PFS.

The objective function thus contains two parts. The separable part of objective 1 is the per-galaxy success rate of redshift measurements. The success rate, a number between 0 and 1, is calculated by fitting the simulated noisy spectrum of the galaxy, and inferring of a redshift can be estimated from the spectrum with the desired precision. We use the same galaxy simulation as in \cite{Cranmer2021-jd}, which employs a single spectral type for every galaxy, so that the redshift success is a function of redshift, mass, and exposure time only. 
We calculate the success rate $\mathrm{SR}_i(t)$ of galaxy $i$ after $t=1,\dots 4$ exposures, and then linearly interpolate them:
\eqsplit{
\label{eq:case2-fi}
f_i(\tau_i)=
\begin{cases}
\tau_i \mathrm{SR}_i(1), &0\leq\tau_i\leq1\\
(\tau_i-1) (\mathrm{SR}_i(2)-\mathrm{SR}_i(1))+\mathrm{SR}_i(1), &1\leq\tau_i\leq2\\
(\tau_i-2) (\mathrm{SR}_i(3)-\mathrm{SR}_i(2))+\mathrm{SR}_i(2), &2\leq\tau_i\leq3\\
(\tau_i-3) (\mathrm{SR}_i(4)-\mathrm{SR}_i(3))+\mathrm{SR}_i(3), &3\leq\tau_i\leq4\\
\end{cases}
}

The non-separable part for objective 2 amounts to counting the number of galaxies that satisfy the specified redshift and mass requirements and that are observed by at least one exposure. Let $\Theta$ be the set of all such galaxies. We adopt the following continuous approximation:
\begin{equation}
\label{eq:case2-s}
    s(\tau_1,\dots,\tau_I)=10000\, \sigma\left(\frac{n-5000}{100}\right)\ \mathrm{with}\ n\equiv\sum_{i\in\Theta}\sigma\left(\frac{\tau_i-0.5}{0.2}\right),
\end{equation}
i.e. $n$ denotes the number of observed galaxies satisfying the selection requirements. This objective term prefers $n>5000$ and is saturated at $n\approx 5500$. The prefactor 10000 is a large number compared to $\sum_i f_i$, chosen to ensure that the second objective receives preference over the first. This choice needs to be made for any multi-objective optimization. The sharpness of the sigmoid functions, 0.2 and 100 in \autoref{eq:case2-s}, are two hyperparameters. Larger sharpness leads to a better approximation to the step function, but is also more difficult to optimize. One could start with small sharpness parameters and then gradually increase them during the training, but we achieve good results with fixed parameters after a hyperparameter search. 

We chose the $\ell_2$ penalty function for case 2 to reduce over- and under-allocation. In contrast to case 1, the time allocation is strongly limited and insufficient to saturate both objectives for the large number of available galaxies. We expect that the under-allocation penalty will become largely obsolete at the end of training but that it provides more meaningful gradient directions during training.

\section{Feature Sets and Training}
\label{sec:training}

Of particular importance are the feature set for the items, which in our cases correspond to one galaxy per node. We thus need to provide to the initial item nodes all features that meaningfully describe the optimization problem from the perspective of the galaxies.

In Case 1, the feature set comprises $T_m$, an one-hot version of the class index from \autoref{tab:classes}, and an extra random number, which distinguishes between different galaxies in the same class.
All other nodes, edges and global features of the graph are initialized with zeros. 
In Case 2, the item node features are initialized to ($\mathrm{SR}_i(1),\mathrm{SR}_i(2),\mathrm{SR}_i(3),\mathrm{SR}_i(4)$ of \autoref{eq:case2-fi}) and a Boolean variable showing whether or not the galaxy satisfies the redshift and mass requirements of \autoref{eq:case2-s}), while all other nodes, edges and global features of the graph are initialized with zeros.

In both cases, we use 10 graphs to train the GNN model, 5 graphs to validate and 5 graphs to test its behavior. There is no overlapping region between training, validating and testing graphs. The model is trained with Adam \citep{Kingma2015-pq} on a 320 NVIDIA P100 GPU. We start with a 2000-epoch pre-training phase with a fixed penalty strength $\lambda$,  followed by a 8000-epoch training with exponentially increasing $\lambda$. Other training parameters are shown in \autoref{tab:train_parm}. 

We do a coarse hyperparameter search over the learning rate, the penalty factor, and the noise level of the noisy sigmoid function. The learning rates in both cases are searched from $10^{-4}$ to $10^{-2}$. The penalty factor in Case 1 is varied between $10^{-8}$ and $10^{-6}$, in Case 2 between $10^{-3}$ and $10$. And the noise level is searched between $0.1$ and $0.4$. The sharpness of the noisy sigmoid method is fixed to 20. The dimensionality of the GNN functions $n_x$, $n_h$ and $n_u$
is set to 10, while $n_g$ is set according to the item features listed above. We experimented with 20-dimensional features but found no improvements.

\begin{table}[]
\caption{Training Parameters. LR is the learning rate, $\lambda$ is the penalty factor and $l$ is the noise level.}
\centering
\label{tab:train_parm}
\begin{tabular}{cccccccc}
\hline\hline
       & \multicolumn{3}{c}{Pre-Training}    & \multicolumn{4}{c}{Training}                            \\
       \hline
       & LR               & $\lambda$        & $l$  & LR               & $\lambda_{start}$ & $\lambda_{end}$  & $l$ \\
       \hline
Case 1 & $5\times10^{-4}$ & $1\times10^{-7}$ &0.3 & $5\times10^{-4}$ & $1\times10^{-7}$  & $1\times10^{-4}$ &0.3 \\
Case 2 & $1\times10^{-3}$ & 0.1              &0.3 & $1\times10^{-3}$ & 0.1               & 1.0              &0.3 \\
\hline\hline
\end{tabular}
\end{table}

\section{Results}
\label{sec:result}

We report the GNN test scores in \autoref{tab:result1} and \autoref{tab:result2} in terms of the objective function as well as the adherence to the constraints.
To the latter end, we define the total overtime and unused time, i.e., $\Delta T=\sum_k\max(0,\sum_{i\in \Psi_k}t_{ik}-T)$ and $\Delta T'=\sum_k\max(0,T-\sum_{i\in \Psi_k}t_{ik})$, and calculate the fraction of such over/unused time compared to the total available observation time $T_\mathrm{all}=T\,K$. The result is written as ${f_0}^{+\Delta T/T_\mathrm{all}}_{-\Delta T'/T_\mathrm{all}}$.
For example, $10^{+3\%}_{-2\%}$ means that the value of the objective function is 10, with 3\% overtime and 2\% unused time.

\subsection{Case 1: Balancing Predefined Classes}

The training and test data was derived from a galaxy catalog provided by the PFS galaxy evolution program.
For classes 1-8 in \autoref{tab:classes}, we use the EL-COSMOS catalog \citep{Saito2020}, which is based on the COSMOS2015 photometric catalog \citep{Laigle2016}.
Since the area coverage of this catalog is too small for simulations of multiple PFS pointings, we repeat the central region of the catalog in a $3\times 3$ tiling pattern, so that the final extended catalog covers a contiguous area of $\sim 10\, \mathrm{deg}^2$.
The remaining classes are artificially superposed on the same region so that the number densities are consistent with the expectation.
Each galaxy in the catalog has a label indicating the class it belongs to. 

We compare our GNN approach to the currently employed network flow optimization method, which is based on the fiber-assignment method in \citet{Blanton2003-ni}. Similar to our approach, it constructs a graph connecting fibers and galaxies, but then solves a \red{linear} min-cost max-flow problem on the graph with the MIP optimizer \texttt{GUROBI}, given predetermined costs for every galaxy class: $f(\tau_1,\dots,\tau_I) = \sum_i c_m\, \iota(i \in \mathcal{C}_m \wedge \tau_i \geq T_m)$, where $\iota$ denotes the indicator function.
Multi-exposure programs like case 1 can be implemented by creating a graph with one fiber node per exposure.
The network flow optimization guarantees feasibility but does not permit the adjustment of the class costs to maximize the objective function. We therefore adopt, \red{as a baseline and a representation of the current state of development, the fixed costs $c_m$ from \autoref{tab:classes} which were identified through manual exploration of the linear objective listed above. It is important to emphasize that these costs they were determined with the same general goal, namely to achieve an equitable distribution of completeness across all galaxy classes, but not the specific objective function in \autoref{equ:case1-obj}.}

For a more flexible optimization of the objective function, we also solve the problem of \autoref{equ:case1-obj} in the form of \autoref{equ:opt}, i.e. directly for $\mathcal{O}(10^5)$ of $t_{ik}$, by ordinary gradient descent.
We use \autoref{eq:noisy-sigmoid} to convert $t_{ij}$ to integers at test time.
We have tried different types of gradient descent (Adam, momentum), but the results are very similar.


\begin{table}[]
\caption{Case 1 results in terms of the values of the objective function in \autoref{equ:case1-obj} (minimal completeness across the classes in \autoref{tab:classes}) from network flow optimization with preset costs (`Baseline-LP`); direct gradient descent of \autoref{equ:opt} (`GD'); and our GNN method for 5 independent test fields. The percentages denote the fraction of the full time allocation $T$ that is overallocated ($^+$) or underallocated ($_-$), averaged over all fibers.}
\label{tab:result1}
\centering
\begin{tabular}{cccc}
\hline\hline
Field ID & \red{Baseline-LP} & GD & GNN (Ours)\\
\hline
1        &$0.773_{-19.3\%}^{+0.0\%}$             &$0.824_{-1.4\%}^{+0.8\%}$    &$0.877_{-9.9\%}^{+0.1\%}$     \\
2        &$0.764_{-20.1\%}^{+0.0\%}$             &$0.827_{-1.6\%}^{+0.8\%}$    &$0.876_{-10.2\%}^{+0.1\%}$     \\
3        &$0.767_{-20.5\%}^{+0.0\%}$             &$0.829_{-1.8\%}^{+0.8\%}$    &$0.880_{-10.5\%}^{+0.1\%}$     \\
4        &$0.768_{-20.6\%}^{+0.0\%}$             &$0.828_{-2.0\%}^{+0.8\%}$    &$0.870_{-10.7\%}^{+0.1\%}$     \\
5        &$0.775_{-20.7\%}^{+0.0\%}$             &$0.830_{-1.9\%}^{+0.8\%}$    &$0.871_{-10.8\%}^{+0.1\%}$    \\
\hline\hline
\end{tabular}\\
\end{table}

The results are shown in \autoref{tab:result1}. In all 5 test fields, our GNN method outperforms the current baseline and the gradient descent solver despite being trained on fields different from the test fields.
The network-flow fiber assignment provides a good baseline with a minimum completeness of $\approx76\%$, but it leaves $\approx 20\%$ of the time unallocated. \red{This apparent contradiction is not an indication of suboptimal performance of the method itself. Instead,} it suggests that the pre-determined class costs of \autoref{tab:classes} are suboptimal for this specific objective function.
The GD method, which like our GNN optimizes \autoref{equ:case1-obj}, improves upon this baseline. 
But we find that, depending on the initialization, it can require a very large number of iterations to converge to a (local) minimum, 
as expected for such a high-dimensional optimization problem.
The GNN benefits from learning a model of what makes galaxies valuable in relation to the constraints, and it communicates that through message passing on the graph.
While the GNN MLPs have in total $O(10^4)$ parameters themselves, they encode the \emph{strategy} of solving \autoref{equ:opt} with galaxy and fiber configurations as given by the training data and the instrument. As a result, similar galaxies will generally be evaluated similarly. This generalization leads to an increased completeness of $\approx 88\%$ even though the solution has not been optimized on the test fields.

With respect to feasibility, unused time is not a concern for case 1. We expected that conflicts between highly valuable galaxies will prevent full utilization of the time allocation, and have confirmed that in the test results. For instance, a canonical problem arises from multiple long-integration galaxies being located in the patrol region of a single fiber. Because of the partial overlap of the patrol regions, some, but not all, of these conflicts can be solved by utilizing a neighboring fiber. If that cannot be achieved, a fraction of the available time cannot be used to increase the completeness of the respective class and, in turn, of the objective function.
However, in comparison to Baseline, the GNN approach evidently converts unused time into gains of the objective, which reveals the suboptimality of having to predetermine the costs for this complex resource allocation problem.
Interestingly, GD does not achieve higher completeness than the GNN despite utilizing almost all the available time.

Overtime violations are, by design, impossible for the network flow method, and are almost completely avoided by the GNN strategy.
As we detail in \autoref{sec:results-case2}, a minor overtime violation is acceptable in this case, but could be avoided entirely by increasing the penalty strength beyond the final value in \autoref{tab:train_parm}.
The Brute Force solver has minor overtime allocations, smaller than the unused allocations, consistent with the asymmetry of the penalty.

In addition to the highest objective function values, GNN is also the fastest method. For every field, both Baseline and GD need to be run again, while the runtime of the GNN is less than 1 second once the training is done. However, even if we include the training time, the GNN is still faster than a single run of the network flow optimization with {\tt GUROBI}.

\subsection{Case 2: Optimizing a General Objective Function}
\label{sec:results-case2}

The training and test data were derived from {\it UniverseMachine} simulations \citep{Behroozi19}, which has a size of $4.0\times4.0$\,deg$^2$, comprising about 35,000 galaxies in a single PFS field of view. The spectrum simulation follows the approach in \citet{Cranmer2021-jd}, which uses a single spectral type of a massive elliptical galaxy, artificially redshifted, and scaled in amplitude to match the expected performance of PFS for a given stellar mass. Stellar masses were predicted from {\it UniverseMachine} halo masses according to the scaling relation in \citet{Girelli2020-ef}. The precision of the redshift estimates was determined by fitting the known spectrum template to 100,000 such galaxy spectra in the presence a constant sky spectrum and the corresponding Poisson shot noise. This procedure constitutes a best-case scenario because spectral misclassification is impossible and catastrophic outliers are rare.

Case 2 again cannot directly be solved with \red{LP} techniques because the main aspect of this problem lies in the determination of the relative importance of the two competing objectives as well as the individual per-galaxy utilities of objective 1 (the precision of the redshift estimation).
We therefore adapt a known heuristic approach to precondition the problem, \red{so that we can express it as a LP problem}.
We first randomly select 5,000 galaxies satisfying the redshift and halo mass conditions and label all of these galaxies as class 1, to be observed with a single exposure. Giving this class infinite costs ensures to saturate \autoref{eq:case2-s}.
We then chose a proposed time allocation $\tau_i$ for all other galaxies $i=1,\dots,N$, so that it maximizes the expected gain, $\tau_i = \mathrm{argmax}_{\tau\in\{0,1,2,..,T\}} \left[\tfrac{f_i(\tau)}{\tau}\right]$ \citep{Dantzig1957-ca}, where $f_i$ is defined in \autoref{eq:case2-fi}.
\red{The same min-cost max-flow MIP} solver we used for case 1 is then run with $1+N$ classes, where $N$ classes are comprised of only one galaxy each and specified by their proposed time and expected utility $f_i(\tau_i)$. Because the classes are defined separately for the two objectives, galaxies in class 1 cannot be used for redshift measurement, necessarily leading to a suboptimal solutions for galaxies that are useful for both objectives.
We also run the brute-force Gradient Descent method for comparison.

The results are shown in \autoref{tab:result2}. 
Because the second objective term $s$ is saturated in all cases, we only show the total redshift success rate of \autoref{eq:case2-fi} as the objective. 
We can see that the results of the GNN method are superior to GD and the \red{Baseline} method in terms of the objective function. This result demonstrates that our method is capable of finding effective strategies for allocating resources in this general test case that combines separable and non-separable objectives.

We note that the GD method is closer to the GNN results than it was in case 1, which we attribute to the reduced volume of the parameter space due to the shorter program times ($T_\mathrm{max} = 4$ instead of $T_\mathrm{max} = 15$).
We also find that the GNN method yields mild levels of feasibility violations.
Although we could in principle avoid such violations by further increasing the penalty factor $\lambda$, we allow them here because observations with PFS will simultaneously allocate about 10--20\% of the fibers as calibration targets.
We decided to ignore this operational complication for this work, but, because the numbers of calibrations measurements are flexible, we can compensate a small amount of over- or unused time with the calibration allocations. 

\begin{table*}[]
\centering
\caption{Case 2 test results in terms of the values of the objective function in \autoref{eq:case2-fi} (i.e. aggregated redshift success rate; the second objective of \autoref{eq:case2-s} is fully saturated by design) from three competing strategies for 5 independent test fields. The percentages denote the fraction of the full time allocation $T$ that is overallocated ($^+$) or underallocated ($_-$), averaged over all fibers.}
\label{tab:result2}
\begin{tabular}{cccc}
\hline\hline
Field ID & \red{Baseline-LP} & GD & {GNN (Ours)} \\
\hline
1                                &$2184.7_{-0.0\%}^{+0.0\%}$                        &$2485.6_{-0.0\%}^{+0.0\%}$                      &$2593.1_{-0.4\%}^{+1.2\%}$          \\
2                                &$2084.4_{-0.0\%}^{+0.0\%}$                            &$2404.2_{-0.0\%}^{+0.0\%}$                      &$2485.6_{-0.5\%}^{+1.1\%}$            \\
3                                &$2151.7_{-0.0\%}^{+0.0\%}$                            &$2457.1_{-0.0\%}^{+0.0\%}$                      &$2544.4_{-0.4\%}^{+1.1\%}$            \\
4                                &$2295.6_{-0.0\%}^{+0.0\%}$                            &$2590.4_{-0.0\%}^{+0.0\%}$                      &$2696.1_{-0.5\%}^{+0.9\%}$            \\
5                                &$2308.5_{-0.0\%}^{+0.0\%}$                            &$2623.6_{-0.0\%}^{+0.0\%}$                      &$2711.9_{-0.5\%}^{+1.0\%}$            \\
\hline\hline
\end{tabular}\\
\end{table*}

\section{Summary and Outlook}
\label{sec:summary}

Resource allocation problems arise in many application areas but remain challenging, especially if they involve high-dimensional and discrete allocation spaces \red{and non-linear or non-separable objectives}. In this paper we present a bipartite GNN architecture that learns a strategy for solving general resource allocation problems. 
It is based on message passing on a graph formed from nodes representing the items of value and the allocation constraints, respectively, connected by edges corresponding to all possible allocations.
It is trained to minimize any user-specified objective function, augmented by a penalty for constraint violations, using instances of the problem -- either from historical occurrences or simulations -- that should capture all relevant aspects of the problem at test time.

We apply our GNN method to the target selection problem in astronomy, which, when given a total observing time budget, amounts to choosing which celestial sources from within a given sky area are to be observed, and for how long. Specializing on a highly multiplexed instrument, the Prime Focus Spectrograph for the Subaru Telescope at Maunakea in Hawai`i, results in the additional complication of having to assign discrete and identical exposure times to sources observed simultaneously by all 2,394 fibers of this instrument.

We demonstrate that our GNN method finds efficient allocation strategies in two realistic problem settings \red{with non-linear and non-separable objectives}.
We compare our results to two direct solvers, one performing a minimum-cost maximum-flow network optimization \red{with predetermined costs}, and the other directly solves for all possible allocations by gradient descent.
Our method yields higher values of the objective function in all cases for every test field.
It formally guarantees feasibility only for infinitely large penalties, and we recommend to increase the penalty term during training until feasibility is achieved or feasibility violations are deemed tolerable. The tuning of the feasibility penalty also allows the exploration of strategies in systems with some amount of slack or surplus, as we expect in the case of PFS.

The development of this GNN method for resource allocations bring two important benefits for future work.
First, the runtime for the GNN solution is much shorter than that of direct solvers, of order 1 second compared to several hours in some cases.
\red{Substantial accelerations by neural MIP solvers have also been found in \citet{Nair2020-tv}. In our case, performing the GNN optimization to precondition a traditional MIP solver should lead to substantially reduced computational costs while maintaining the guaranteed feasibility of that solver.}
Either option will render it practically doable to roll out strategy updates over a large number of problem instances or to assess the probabilities that any item receives some amount of allocation. This so-called `selection function' is of critical importance for precision analyses in astrophysics and cosmology.

Second, multi-objective problems require the balancing of priorities for different kinds of items (e.g. galaxies in our case 1), which traditionally have to be established beforehand. If the respective utilities are not known a priori, as is routinely the case in scientific experiments, the complexity of this task renders it unlikely that manual exploration of the priorities yield near-optimal results. 
Our GNN provides a differentiable architecture, thereby exposing all relevant parameters of the problem to optimization. 
Similar to \citet{Cranmer2021-jd}, we intend to make use of this capability in forthcoming works to train another neural network to learn the utility of galaxies based on easily observable features instead of assuming that these utilities are known, as we have done in test case 2.

The permutation invariance and flexible node and edge models of GNNs render them exceptionally well suited for resource allocation problems. We suspect that is should also work well e.g. for auction strategies \citep{Huang2008-xp}. Other interesting questions beyond the scope of this work relate to the goal of Explainable AI, for instance: what information is passed between the nodes of the graph; how many message-passing steps are needed to achieve these results; and what role does the global model play. 

The GNN code used in this paper is available at  \url{https://github.com/tianshu-wang/PFS-GNN-bipartite}.

\section*{Acknowledgements}
The authors want to thank Kiyoto Yabe for his help with application case 1.
This work was supported by the AI Accelerator program of the Schmidt Futures Foundation.

\bibliography{example_paper}

\appendix

\section{Proof of the Theorem}
\label{sec:proof}

Let $V$ be the set of all vertices and $E$ be the set of all edges, we have a hypergraph $G=(V,E)$. The connectivity of this graph is represented by the incidence matrix $\tA\in\mathbb{R}^{|V|\times|E|}$, where $\tA_{ij}=1$ if and only if edge $j$ is connected to vertex $i$, otherwise $\tA_{ij}=0$. 

The time allocations $t_{ikl}$ between galaxy $i$ and fiber $k$ in exposure $l$ from \autoref{equ:decompose} are represented by vectors $\vec{E}_l\in\{0,1\}^{|E|}$ $(l=1,\dots,T)$. The $j$th element of $\vec{E}_l$ equals $t_{ikl}$ if the $j$th edge in $E$ connects item node $i$ and constraint node $k$. Similarly, $t_{ik}$ can be represented by a vector $\vec{E}_{tot}\in\{0,1,2,...,T\}^{|E|}$, and the $j$th element of $\vec{E}_{tot}$ equals $t_{ik}$.
The target selection problem \autoref{equ:opt} is then written as
\eqsplit{
\label{equ:total}
&\arg\max_{\vec{E}_{tot}}f(\vec{E}_{tot})\\
&\tA\cdot \vec{E}_{tot}\leq T\vec{1}_{|V|}\\
&\vec{E}_{tot}\in\{0,1,2,...,T\}^{|E|}
}
We want to decompose $\vec{E}_{tot}$ into a set of $\vec{E}_l$ that satisfy
\eqsplit{
\label{equ:per-step}
&\vec{E}_{tot}=\sum_{l=1}^T\vec{E}_l\\
&\tA\cdot \vec{E}_l\leq \vec{1}_{|V|}\\
&\vec{E}_l\in\{0,1\}^{|E|}
}

\begin{theorem}
Given a solution $\vec{E}_{tot}$ to the problem \autoref{equ:total}, there exists at least one set $\{\vec{E}_1,\dots,\vec{E}_T\}$ satisfying \autoref{equ:per-step}.
\end{theorem}

Proof by induction. When $T=1$, $\vec{E}_{tot}=\vec{E}_1$ and the theorem holds trivially.
Assume that the statement is true for $T=S$. For $T=S+1$, we have $\vec{E}_{tot,S+1}$ which satisfies
\eqsplit{
&\tA\cdot \vec{E}_{tot,S+1}\leq (S+1)\vec{1}_{|V|}\\
&\vec{E}_{tot,S+1}\in\{0,1,2,...,S+1\}^{|E|}
}
If we can find $\vec{E}\in\{0,1\}^{|E|}$ such that $\tA\cdot (\vec{E}_{tot,S+1}-\vec{E})\leq S\vec{1}_{|V|}$ and $\vec{E}_{tot,S+1}-\vec{E}\geq\vec{0}$, the problem is converted to a $T=S$ problem and we can thus find a subset \{$E_1$,...,$E_S$\}$\subset \{0,1\}^{|E|}$ such that $\vec{E}_{tot}-\vec{E}=\sum_l\vec{E}_l$. Combining all $\vec{E}_l$ and $\vec{E}$ gives a decomposition of $\vec{E}_{tot}$. Thus the theorem is equivalent to the existence of such $\vec{E}$.

$\vec{E}$ is given by the following problem:
\eqsplit{
&\tA\cdot (\vec{E}_{tot,S+1}-\vec{E})\leq S\vec{1}_{|V|}\\
&\vec{E}_{tot,S+1}-\vec{E}\geq0\\
&\vec{E}\in\{0,1\}^{|E|}
}
where $\vec{E}_{tot,S+1}$ satisfies $\tA\cdot \vec{E}_{tot,S+1}\leq (S+1)\vec{1}_{|V|}$ and $\vec{E}_{tot,S+1}\in\{0,1,2,...,S+1\}^{|E|}$.

Let $\vec{A}_v$ be the $v$th row of $A$ and $U=\{v|\vec{A}_v\cdot\vec{E}_{tot,S+1}=S+1\}$. For any $v\in U$, we must have $\vec{A}_v\cdot\vec{E}=1$. The problem becomes
\eqsplit{
\label{equ:induction}
&\tA_U\vec{E}=\vec{1}_{|U|}\\
&\tA_{V/U}\vec{E}\leq\vec{1}_{|V/U|}\\
&\vec{E}_{tot,S+1}-\vec{E}\geq0\\
&\vec{E}\in\{0,1\}^{|E|}
}
We generalize the problem into a linear system so that $\vec{E}$ can take any number between 0 and 1:  
\eqsplit{
\label{equ:linear}
&\tA_U\vec{E}=\vec{1}_{|U|}\\
&\tA_{V/U}\vec{E}\leq\vec{1}_{|V/U|}\\
&\vec{E}\leq \vec{E}_{tot,S+1}\\
&\vec{0}_{|E|}\leq\vec{E}\leq\vec{1}_{|E|}
}
Therefore, the theorem is equivalent to the existence of integer solutions of \autoref{equ:linear}. The existence of such integer solutions is guranteed by the following lemma:
\begin{lemma}
The solution set of problem \autoref{equ:linear}, a convex polytope, contains at least one integer point. 
\end{lemma}

\begin{proof}
First, we can show that this solution set is not empty because $\vec{E}'=\frac{1}{S+1}\vec{E}_{tot,S+1}$ is obviously a solution.
Now, consider an arbitrary corner of this polytope, $\vec{E}^\star$. The corner is determined by $|E|$ linearly independent equations. Equations come from the bottom two conditions will directly give the value of the corresponding element in $\vec{E}^\star$. The remaining undetermined elements of $\vec{E}^\star$ is then determined by the first two conditions, i.e., by the linear equations defined by $\tA'$, an invertible square submatrix of $\tA$. 
Since the graph is bipartite, $A$ is totally unimodular. This means that any square submatrix has determinant 1, 0 or -1. Because the submatrix $\tA'$ is invertible, its determinant can only be $\pm1$. Then by Cramer's rule, the inverse matrix is also an integral matrix. Thus the solution to the linear equations, the undetermined elements of $\vec{E}^\star$ are integers. Therefore, any corner of the solution set is an integer point. And because the set is non-empty, there must be at least one corner $\vec{E}^\star$ which is the solution to problem \autoref{equ:induction}. 
\end{proof}

\subsection*{Time Complexity}
To find such a decomposition, we can find a sequence of $\vec{E}_l$ by recursively finding $\vec{E}^\star$. Finding $\vec{E}^\star$ is no slower than polynomial time, because we can randomly choose a vector $\vec{c}$ and maximize $\vec{c}\cdot\vec{E}$ within the polytope. Since the linear programming problems can be solved in polynomial time, finding $\vec{E}^\star$ and the sequence $\{\vec{E}_l\}$ can also be done in polynomial time.

\end{document}